\documentclass[%
 reprint,
 amsmath,amssymb,
 aps,
prl,
floatfix
]{revtex4-1}
\usepackage{amsmath}  
\usepackage{amsfonts}
\usepackage{amssymb}
\usepackage{amsthm}
\usepackage[]{nicefrac}
\usepackage{bbm}
\usepackage{qcircuit}
\usepackage{color}
\usepackage{float}
\usepackage[normalem]{ulem}
\usepackage{cancel}
\usepackage{graphicx}
\usepackage{varwidth}
\usepackage{enumitem} 
\usepackage{scalerel}  

\newcommand{\ignore}[1]{}
\newcommand{\ket}[1]{\left|#1\right\rangle}
\newcommand{\bra}[1]{\left\langle#1\right|}
\newcommand{\braket}[2]{\left\langle #1| #2 \right\rangle}
\newcommand{\ketbra}[2]{ \left| #1 \right\rangle\left\langle #2 \right|}
\newcommand{\prnt}[1]{\left( #1 \right)}
\newcommand{\prntt}[1]{\left[ #1 \right]}\newcommand{\prnttt}[1]{\left\{ #1 \right\}}
\newcommand{\abs}[1]{\left| #1 \right|}
\newcommand\norm[1]{\left\lVert#1\right\rVert}

\DeclareMathOperator{\tr}{tr}

\def\th{\ensuremath{^\mathrm{th}} }

\def\poly{\ensuremath{\textrm{poly}}}

\newtheorem{theorem}{Theorem}
\newtheorem{corollary}{Corollary}
\theoremstyle{definition}
\newtheorem{definition}{Definition}
\newtheorem{lemma}{Lemma}
\newtheorem{fact}{Fact}
\newtheorem{claim}{Claim}

\begin{document}
\title{How the High-energy Part of the Spectrum Affects the Adiabatic Computation Gap}
\author{Yosi Atia}
\email{g.yosiat@gmail.com}

\author{Dorit Aharonov}

\affiliation{School of Computer Science and Engineering,
The Hebrew University of Jerusalem,
The Edmond J. Safra Campus,
9190416 Jerusalem, Israel}

\begin{abstract}
Towards better understanding of how to design 
efficient adiabatic quantum algorithms, we 
study how the adiabatic gap 
depends on the spectra of the initial and final Hamiltonians in a natural family of test-bed examples. We show that  perhaps counter-intuitively, changing the energy in the initial and final Hamiltonians of only highly excited states (we do this by assigning all eigenstates above a certain cutoff the same value), can turn the adiabatic algorithm from being successful to failing. Interestingly, our system exhibits a phase transition; when the cutoff in the spectrum becomes smaller than roughly $n/2$, $n$ being the number of qubits, the behavior transitions from a successful adiabatic process to a failed one. To analyze this behavior, 
and provide an upper bound on both the minimal gap as well 
as the success of the adiabatic algorithm, we introduce the notion of \textit{escape rate} which quantifies the rate by which the system escapes the initial ground state (a related notion was also used in \cite{IL18}).  
Our results indicate a phenomenon that is interesting on its own right: an adiabatic evolution may be robust to bounded-rank perturbations, even when the latter closes the gap or makes it exponentially small.

\end{abstract}
\maketitle
\bibliographystyle{apsrev4-1}
Adiabatic quantum computation~\cite{FGGS00, AL18} is an alternative paradigm to the circuit-model quantum computation.  The system is initiated to a known, easy to prepare ground state of a Hamiltonian $H_0$, and the solution to the problem is encoded in the ground state of a Hamiltonian $H_1$. Then, the system evolves by a time-dependent Hamiltonian, which slowly varies from $H_0$ to $H_1$. By the adiabatic theorem \cite{Messiah}, if the evolution is slow enough, the system remains close to the instantaneous ground state throughout the evolution, and eventually reaches the ground state of $H_1$. 

Although the adiabatic computation model is equivalent to the quantum circuit model~\cite{ADKLLR07}, only a handful of adiabatic algorithms that do not originate from a circuit implementation have emerged \cite{RC02, SNK12, Hen14, SSO18}. An efficient adiabatic algorithm requires an efficiently implementable Hamiltonian evolution with a gap which is at least inverse polynomially small, so that the final ground state can be found. While there is a lot of design freedom in choosing the Hamiltonian evolution, proving that the gap is sufficiently large is usually very difficult (various techniques are reviewed in \cite{BFKSZ13, Baume16, AL18}), and, in general, undecidable  \cite{CPW15}. Previous work have proved the gap is small when the Hamiltonians exhibit localization phenomenons (see \cite{LMSS15}). When the degree of the interaction graph of the Hamiltonian is large, the system may fail to explore the Hilbert space fast enough to follow the ground state, which leads to a small gap \cite{DMV01,FGGN08,JKKM08,FGGGS10}. In more spatially-local models, Anderson localization \cite{Anderson58} and many body localization \cite{NH14},  may similarly close the gap \cite{AC09, AKR10, KS10, LMSS15}. 


In this note we provide some surprising insight about 
the dependence of the success of the adiabatic algorithm, on the 
full spectra of the initial and final Hamiltonian \footnote{We use the same notion of success of a quantum adiabatic algorithm as in \cite{FGGN08}, namely, the probability to reach the final ground state is non-neglectable. Note that in other works the notion of success may mean something 
different.}.
We show that this dependence is stronger than 
what might have been expected: the success probability significantly depends on the high-energy parts of the spectrum; in particular, to energies close to $n/2$. 
To study this question, we consider a family of very simple toy-example 
Hamiltonians, and demonstrate the dependence on the high-energy parts of the spectrum by a phase transition our system exhibits. 
Along the way, we develop tools to analyze the success probability of an adiabatic algorithm, by relating it to what we call the {\it escape rate}: the rate by which the system leaves the initial ground state.

To motivate our test-bed Hamiltonians, we start with 
a very simple observation. 
Hereinafter, we assume, 
as is often done in adiabatic algorithms \cite{FGGS00,RC02, ADKLLR07}, that the Hamiltonian evolutions are an interpolation of the initial Hamiltonian $H_0$ and the final Hamiltonian $H_1$ for total time $\tau$, i.e., $H(t)=H_{t/\tau}$, and $H_s=(1-s)H_{0}+s H_{1}$. Consider the following two Hamiltonian evolutions: first, the ``projection problem'' which resembles the adiabatic version of Grover's search algorithm \cite{Grover96,RC02}
\begin{equation} \label{eq:ProjectionProblem}
\begin{split}
H_{0}^{A}&=\mathbbm{1}_{2^{n}}-\ketbra{+ {\dots} +}{+{\dots} +}\\
H_{1}^{A}&=\mathbbm{1}_{2^{n}}-\ketbra{0 \dots 0}{0\dots 0},
\end{split}
\end{equation}
It is an easy calculation to prove that 
the minimal gap of $H^A_s$ is exponentially small; this is because it is an interpolation between two projection Hamiltonians, whose unique ground states have exponentially small inner-product \footnote{See proof in Supplemental Material at [URL will be inserted by publisher]}. 

Consider now also 
a second Hamiltonian evolution:
\begin{equation}
\begin{split}
H_{0}^{B}&=\sum_{k=1}^{n} \ketbra{-}{-}_k=\frac{1}{2}\sum_{k=1}^{n} \prnt{\mathbbm{1}-\sigma_x^k}
\\
H_{1}^{B}&=\sum_{k=1}^{n} \ketbra{1}{1}_k=\frac{1}{2}\sum_{k=1}^{n} \prnt{\mathbbm{1}-\sigma_z^k},
\end{split}
\end{equation}
wherein $k$ subscript terms act  on the $k\th$ qubit.
Again, a very easy calculation shows that 
the minimal gap of $H^B_s$ is a constant, as a sum of single qubit terms \footnote{See proof in Supplemental Material at [URL will be inserted by publisher]}. 

Our starting observation is  that the ground states of $H^A_0$ and $H^B_0$ (denoted $\alpha_0^A, \alpha_0^B$ respectively) are identical, and so are the ground states of $H^A_1$ and $H^B_1$; the difference between the $A$ and $B$  adiabatic evolutions (namely, $H^{A}(t)$ and $H^{B}(t)$) stems from the excited parts of the spectra in their corresponding initial and final Hamiltonians. 

This trivial example already reveals that the higher parts of the energy spectra of the initial and final Hamiltonians have sufficient influence on the gap to change it from a constant to exponentially small.


Motivated by the desire to quantify this difference and dependence on the higher parts of the spectra, we define and analyze a family of Hamiltonian evolutions, which interpolate between these two evolutions, and which are controlled by a single parameter $\theta$:
\begin{equation}
\label{eq:ThetaH}
\begin{split}
H_0^{\theta}&=\sum_{j=0}^{2^{n}-1} h_\theta(j)|j_{+}\rangle\langle j_{+}| ~~~~
H_1^{\theta}=\sum_{j=0}^{2^{n}-1}{h}_\theta(j)|j\rangle\langle j|,
\end{split}
\end{equation}
wherein $h_\theta(j)=\min(\theta,h(j))$, and $h(j)$ is the number of 1s of the binary representation of $j$ (i.e., the Hamming weight). $\ket{j_+}$ is the state $\ket{j}$ after a Hadamard gate ($\ket 0 \rightarrow \ket + , \ket 1\rightarrow \ket -$) is applied to every qubit. Indeed, the cases $H^{1}$ and $H^{n}$ correspond to $H^A$ and $H^B$ respectively. Thus, varying $\theta$ interpolates between the two extreme cases. By gradually changing the value of $\theta$ we can control the energy landscape, and thus, the gap.

We note that while the 
intermediate Hamiltonians 
$H_s^{\theta}$ may be highly 
non local, they are still symmetric to permutations of the $n$ qubits; such Hamiltonians were also studied in \cite{DMV01, FGG02,Reichardt04, BvD16a, BvD16b, MAL16,KC17} (see also \cite{AL18}). 

To  understand how 
the modification of $\theta$ affects the minimal gap along the adiabatic path, note that the initial and final Hamiltonians for $\theta=1$ have a ``flatter'' energy landscape compared with the corresponding initial and final Hamiltonians of $\theta=n$ (see Fig. \ref{fig:landscape}). We can justify the minimal gap difference 
between these two extreme cases using intuition coming from quantum simulated annealing: with a flat energy landscape, the system requires more time to find the direction to the energy minimum, because in the local vicinity there is no ``gradient'' of energy.

\begin{figure} 
\centering
	{\includegraphics[scale=0.19, trim=0.5cm 0cm 1cm 0cm,clip=true] {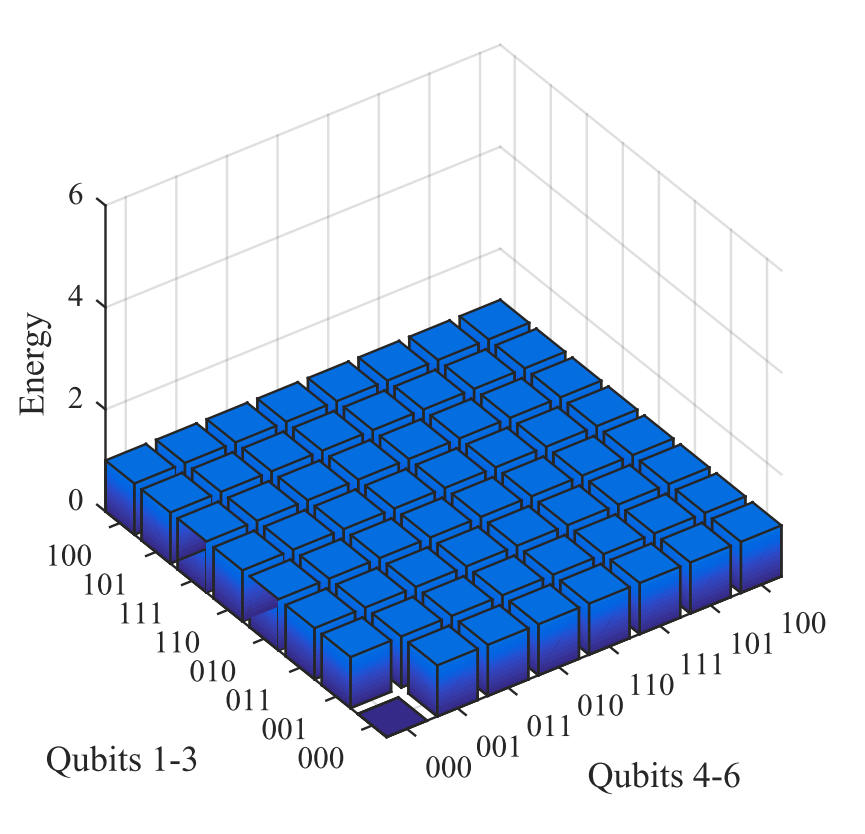}} {\includegraphics[scale=0.18, trim=1.3cm 0cm 1.3cm 0cm,clip=true] {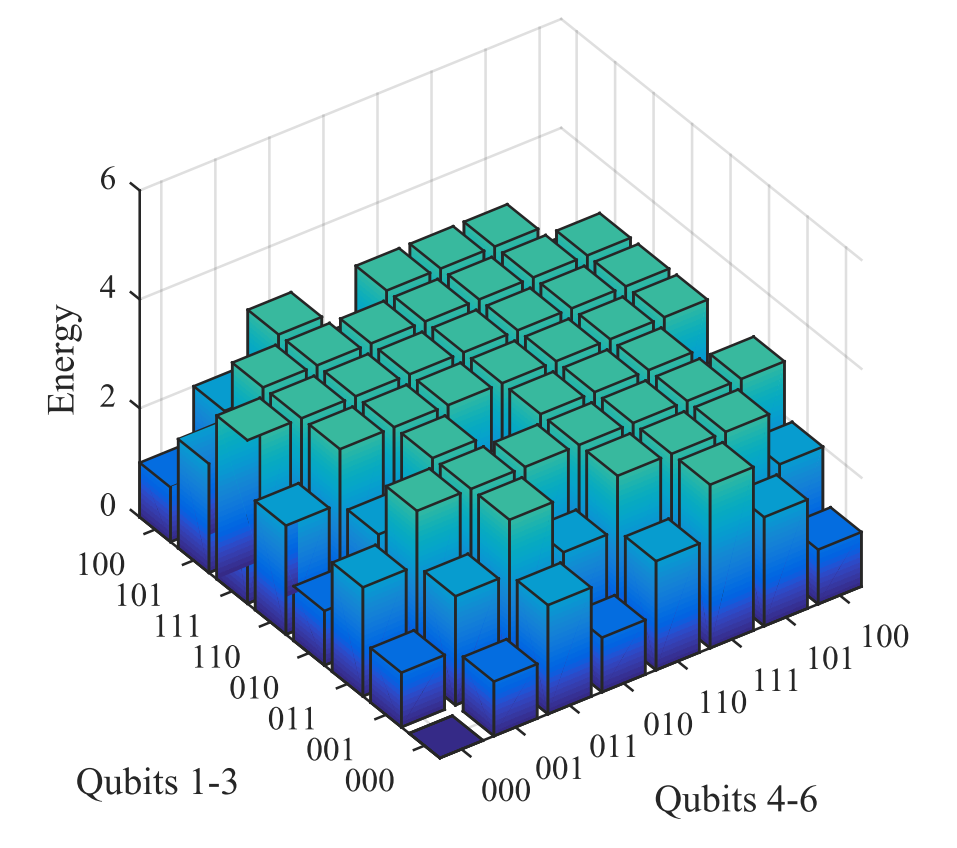}}
	\includegraphics[scale=0.25]{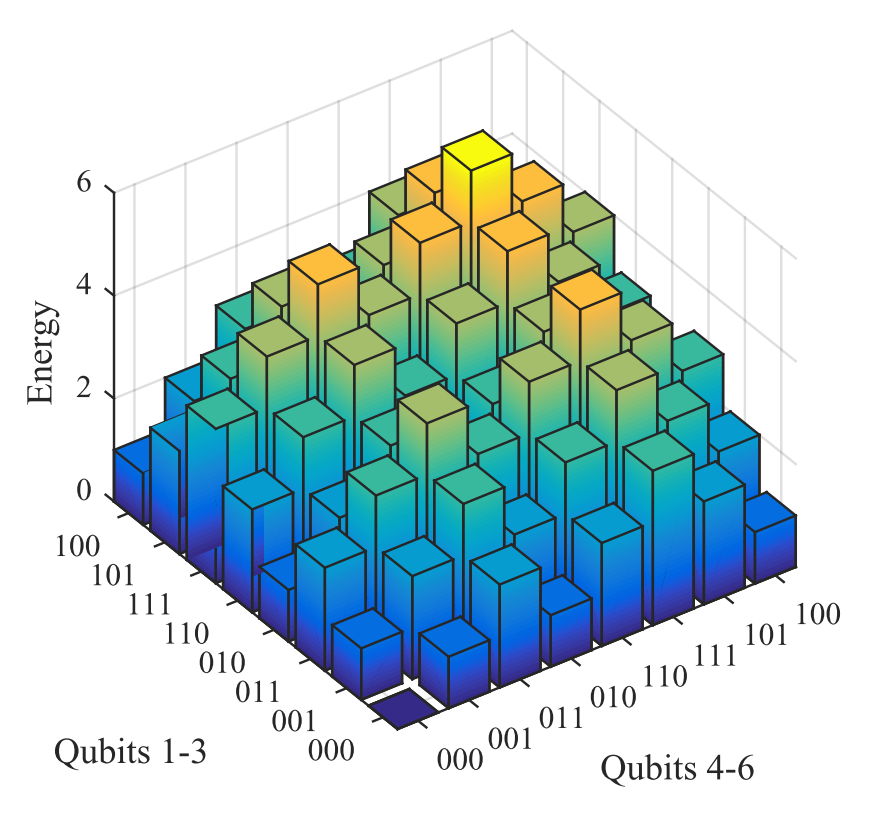}
\caption{\label{fig:landscape} Energy landscape for a 6-qubit $H_1^{\theta}$ with $\theta=1,3,6$. The $xy$-plane spans the computational basis for 6 qubits and the height ($z$ axis) indicates the energy of the state. As $\theta$ decreases, so does the ``energy gradient'', which leads the system to the minimal energy.}
\end{figure}

The question is how flat does the landscape have to be, for the evolution to require a long time, i.e., for the gap to approximately vanish? Intuitively, one would predict that changes in the very high energy levels of the Hamiltonians should not greatly affect the system's evolution, because only low energy states 
affect the adiabatic gap; 
in other words, one might 
naively expect that 
$\theta$ can be fairly small while still maintaining the qualitative behavior of 
$\theta=n$.

{~}

\subsection{Main Result}
It turns out that the  truth is very different than 
perhaps what one might expect at very first sight - a phase transition occurs as perhaps expected, but the critical point only happens at the high value of $\theta \approx n/2$. 

The following theorem proves a phase transition 
    in the success of the adiabatic evolution, as a function of $\theta$ \footnote{See full proof in Supplemental Material at [URL will be inserted by publisher]} and connects it in one case to the minimal gap (asymptotic notations explained in footnote \footnote{The $\Omega,O$ and $o$ notations are often utilized to compare asymptotic functions or series. $f(n)\in O(g(n))$, which is equivalent to $g(n)\in \Omega (f(n))$, means that asymptotically, $f(n)/g(n)$ is at most a positive constant.  $f(n)=o(g(n))$ if $f(n)/g(n)$ converges to 0.}):

\begin{theorem} [\textbf{Main: Phase transition}] \label{thm:main}~ \\
\textbf{a.} Let $\theta_\ell\triangleq \frac{n}{2}-\sqrt{n\log^c n }$ for some constant $c>1$. If $\theta\le \theta_\ell$, then polynomial time adiabatic evolution by $H^{\theta}$ fails.  Furthermore, the minimal gap in this case is $o(1/\poly(n))$ small.\\
\textbf{b.} Let $\theta_h \triangleq n/2+\sqrt{40\log n}$. Evolving by  $H^{\theta}$, wherein $\theta \ge \theta_h$  succeeds  for  $\tau=n^4$.
\end{theorem}

Above we say that the evolution is \emph{successful}, if the final state of the system is with overlap  $\Omega(1/\poly(n))$ with the final Hamiltonian's ground state \footnote{An adiabatic evolution with non-negligible gap, which is run slowly enough, is necessarily successful but the other way around is not necessarily true; cf. \cite{SNK12}}.  
We say the evolution fails if 
the final state is with overlap $o(1/\poly(n))$ with the final ground state. 

To prove the theorem, we relate three properties: the success 
probability of the algorithm, the adiabatic gap and a third notion, 
which we call \emph{escape rate}. 
The escape rate from a subspace $\mathcal V$ by a Hamiltonian $H$ is the rate by which a system at $\mathcal V$ leaves $\mathcal V$ when evolving by $H$ for infinitesimal time.
We prove  Theorem \ref{thm:main}a by showing that for $\theta\le\theta_\ell$,  the escape rate from the ground state of $H_0$ by the Hamiltonian $H_1$ is super-polynomially small,  from which we can deduce the failure the evolution.  By the adiabatic theorem, a failure to reach the final ground state implies that the gap is super-polynomially small (otherwise the final ground state would have been found in polynomial time), and this proves the second part
of Theorem \ref{thm:main}a. 

We note that proving that the gap is small directly 
would not suffice to prove failure of the algorithm;
see e.g. \cite{SNK12}. 

Theorem \ref{thm:main}b shows the other side of the phase transition, namely, that if $\theta\ge 
\theta_h$ the adiabatic algorithm succeeds. We note that we do not know how to prove that the gap in this case is $\Omega(1/\poly(n))$ for $\theta\ge \theta_h$, which would imply the result, though we believe this is true. 
Instead, the result is proven by a study of the robustness 
of the system to a certain type of perturbations. 
More precisely, we show that the path of the system as it evolves by $H^{n}_s$, remains in a subspace which has very little overlap with the subspace spanned by large Hamming-weight states ($\ket{j},\ket{j_+}$, with $h(j)\ge\theta_h$).  Therefore, perturbing their energies has a negligible effect on the path. 
We later extend this tool to a more general statement 
about robustness of a system to 
certain types of perturbations (See Theorem 
\ref{thm:robustnesss}). 

\begin{figure} [h]
\centering
{\includegraphics[scale=0.8, trim=6.7cm 10.2cm 6.7cm 10.3cm,clip=true] {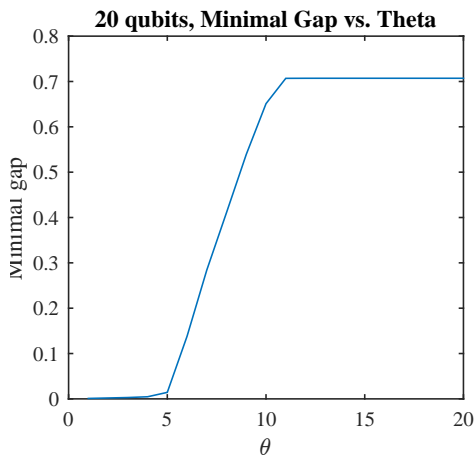}}
\caption{\label{fig:threshold} The minimal spectral gap for 20 qubit $\theta$-Hamiltonian.}
\end{figure}
\subsection{Tools}
We introduce the notion of {\it escape rate} from a subspace. 
\begin{definition} [\textbf{Escape rate}] 
The escape rate from a subspace $\mathcal V$ by Hamiltonian $H, \norm{H}=\poly(n)$, is $\beta$ if 
\begin{equation} \label{eq:NotMoving} 
\begin{split}
\max_{v\in \mathcal V} \norm{\Pi_{\mathcal{V^\perp}} H \ket{v}} = \beta.
\end{split}
\end{equation}
\end{definition}

The escape rate is in some sense a measure of localization: the rate by which a system in $\mathcal V$ can escape $\mathcal V$ in infinitesimal time.  However, note that the escape rate is not necessarily small for Anderson localized systems \cite{Anderson58}, in which the distance a particle can travel on a lattice is limited to a constant, due to disorder. For example, consider a subset of subsequent sites in an Anderson localized particle on a line; this subset of sites defines a subspace $\mathcal V$; in principle,  a particle at the border of the subspace can escape quickly (i.e., large escape rate), though the distance traveled on the lattice is still bounded by some constant. 
Thus, upper bounding the escape rate seems to provide a stronger handle on the behavior of the particle.

We remark that, as we will show later in Eq. \ref{eq:DeltaE}, in case the subspace  $\mathcal V$ is one dimensional, the escape rate from $\mathcal V$ turns out to be exactly equal to the energy uncertainty of the state spanning $\mathcal V$, a quantity which is used in the quantum-speed limit of Mandelstam-Tamm \cite{MT45,Bhattacharyya83}; the connection between the quantum speed limit and  adiabatic evolution was recently explored in \cite{IL18} (see also \cite{AA04} \footnote{Note that the escape rate corresponds to the rate the subspace is abandoned only at the very beginning of the evolution. After that, parts of the system that are outside the space could start flowing back into it, reducing the effective abandoning rate. This is why it is only an upper bound.}).


The following Lemma \ref{lem:EscapeBoundsSuccess} is a tool to rule-out the success of polynomial time adiabatic algorithms, using the notion of escape rate \footnote{See proof in Supplemental Material at [URL will be inserted by publisher]}. We use the $\alpha_s$ to denote the ground state of the Hamiltonian $H_s$.
\begin{lemma} [\textbf{Escape rate bounds success}] \label{lem:EscapeBoundsSuccess} 
Let $H_s$ be an $n$ qubit Hamiltonian evolution by interpolation s.t. $\norm{H_s}=O(\poly(n))$.
 Additionally, let ${\mathcal{V}}$ be a subspace spanned by eigenstates  of $H_0$ whose projection $\Pi_{\mathcal{V}}$ satisfies $\norm{\Pi_{\mathcal{V}}\ket{\alpha_0}}=1$, and $\norm{\Pi_{\mathcal{V}}\ket{\alpha_1}}=o(1/\poly(n))$.  When adiabating $\alpha_0$ from $H_0$ to $H_1$ in time $\tau=O(\poly(n))$, while the escape rate of $\mathcal V$ by $H_1$ is $\beta<\pi/2\tau$, the projection of the final state on $\alpha_1$ is $\le  \sin(\beta\tau)+o(1/\poly(n))$.
\end{lemma}

The idea of the proof is that adiabating  $\alpha_0$, from $H_0$ to $H_1$ can be partitioned to small time-independent Hamiltonian evolutions, alternating between  $H_0$ and $H_1$ by using the Trotter formula \cite{NC00}. Clearly $\alpha_0$ is invariant to $H_0$; hence, the rate by which the amplitude of $\alpha_0$ changes with the evolution can be bounded using its escape rate by $H_1$. If the escape rate is not high enough, the system  remains mostly at $\mathcal V$ and has  negelectable amplitude on $\alpha_1$.

By the adiabatic theorem, if indeed the  system fails to reach the $\alpha_1$ in polynomial time, the minimal gap must have been $o(1/\poly(n))$ small. This argument however gives weak bounds; the following Lemma uses a more subtle argument to provide the upper bound on the gap of Theorem \ref{thm:main}a.  
\begin{lemma} [\textbf{Gap bound by escape rate}] 
\label{lem:stuck_qw_aqc}
Let $H_s$ be an $n$ qubit Hamiltonian evolution s.t. $\norm{H_s}=\poly(n)$. Additionally, let ${\mathcal{V}}$ be a subspace spanned by eigenstates of $H_0$, whose projection $\Pi_{\mathcal{V}}$ satisfies $\norm{\Pi_{\mathcal{V}} \ket{\alpha_0}}=1$  and $\norm{\Pi_{\mathcal{V}}\ket{\alpha_1}}\le 1/10$.
If the escape rate of the subspace $\mathcal V$ by $H_1$ is $\beta$, then 
 $\Delta \le \sqrt[4]{{100\beta}\norm{H_0-H_1}^3}$.
\end{lemma}
 The proof idea is to assume by contradiction that the gap is larger than $\sqrt[4]{{100\beta}\norm{H_0-H_1}^3}$. By the adiabatic theorem, for $\tau=1/\beta$, the system reaches distance $1/10$ from the final ground state. However, $\beta$ is too small to allow such evolution, hence the assumption is false \footnote{See proof in Supplemental Material at [URL will be inserted by publisher]}.

For the proof of Theorem \ref{thm:main}b we use the following simple Lemma, which bounds the influence of a Hamiltonian perturbation on a system's path \footnote{See proof in Supplemental Material at [URL will be inserted by publisher]}. 
\begin{lemma} [\textbf{Path shift by Hamiltonian perturbation}] \label{lem:pert}
Consider two time-dependent Hamiltonians $H(t), \widetilde H(t)$. Let $U(t)$ and $\widetilde U(t)$ be the respective evolution operators by these Hamiltonians  (e.g., when evolving by $H$, $\ket{\psi(t)}=U(t)\ket{\psi(0)}$). Then,
\begin{equation} \label{eq:pertH}
\norm{(U(\tau)-\widetilde U(\tau))\ket{\psi(0)}} \le \intop_0^\tau dt \norm{\prnt{H(t)-\widetilde H(t)}\ket{\psi(t)}}.
\end{equation}
\end{lemma}
In the extreme case, the perturbation is in a subspace which is completely orthogonal to the path of the system and does not change it at all. Lemma \ref{lem:pert} bounds the path's shift  when the path is not in the null-space of the perturbation. Later, we use Lemma \ref{lem:pert} to upper bound the path shift when evolving  $H^{\theta_h}$ in comparison to $H^{n}$,  thus showing that the algorithm succeeds  when perturbing $\approx n/2$ of the top eigenstates..

\subsection{Proof of Theorem \ref{thm:main}}

\textit{Theorem \ref{thm:main}a} is proved by showing that the conditions for Lemma \ref{lem:stuck_qw_aqc} and Lemma \ref{lem:EscapeBoundsSuccess} hold with ${\mathcal{V}}$ being a one dimensional subspace spanned by the ground state of $H_0^{\theta}$, denoted $\alpha^{\theta}_0$, for $\theta\le \theta_\ell$. The overlap condition between the initial and final ground states is $2^{-n/2}$ as required by both Lemmas. It is left to find the value of $\theta$ for which the escape rate of $\mathcal V$ by $H^{\theta}_1$ is $o(1/\poly(n))$ small:   
\begin{equation} \label{eq:escaperatechernoff}
\begin{split}
    &\norm{\Pi_{\mathcal{V^\perp}} H^{{\theta}}_1 \ket{\alpha^{{\theta}}_0}}^2
    \\
    &=\norm{H^{{\theta}}_1 \ket{\alpha^{{\theta}}_0}}^2 - \norm{\Pi_{\mathcal{V}} H^{{\theta}}_1 \ket{\alpha^{{\theta}}_0}}^2
    \\
        &= {\bra{\alpha^{{\theta}}_0}(H_1^{\theta})^2\ket{\alpha^{{\theta}}_0}} - {\bra{\alpha^{{\theta}}_0}H_1^{\theta}\ket{\alpha^{{\theta}}_0}}^2
    \\
    &\le \theta^2 - \prnt{\frac{1}{2^n}{\sum_{k=0}^{n} {n \choose k}  \min(k,\theta)}}^2 
    \\
    &\le \theta^2 - \prnt{\theta - \frac{1}{2^n}\sum_{k<\theta} {n\choose k} (\theta-k)}^2
    \\
    &= \prnt{2\theta - \frac{1}{2^n}\sum_{k<\theta} {n\choose k} (\theta-k)} \prnt{\frac{1}{2^n}\sum_{k<\theta} {n\choose k} (\theta-k)}
    \\
     &\le 2\theta^2 e^{-\frac{(n/2-\theta)^2}{n}},
\end{split}
\end{equation}
where the sum of binomial coefficients was bounded using the Chernoff bound.  By choosing $c>1$ s.t.  $\theta = \frac{n}{2}-\sqrt{n\log^{c} n}$, we get that the escape rate is $o(1/\poly(n))$.  By Lemma \ref{lem:EscapeBoundsSuccess}  the evolution fails, and by Lemma \ref{lem:stuck_qw_aqc} the gap is super-polynomially small. 

As a side remark, note that when $\mathcal V$ is one dimensional the escape rate equals the energy uncertainty of the state $v$ spanning $\mathcal V$:
\begin{equation} \label{eq:DeltaE}
\begin{split}
    \norm{\Pi_{\mathcal V^\perp} H \ket{v}}^2&=\norm{H \ket{v}}^2-\norm{\Pi_{\mathcal V} H \ket{v}}^2
    \\
    &= \bra{v}(H)^2\ket{v}-\bra{v}H\ket{v}^2 = (\mathrm{std}_v H)^2.
\end{split}
\end{equation}
In Eq. \ref{eq:escaperatechernoff}, the escape rate of $\alpha_0^{\theta}$ by $H_1^{\theta}$ upper bounds the energy uncertainty of $\alpha_0^{\theta}$ by $H^{\theta}_s$ for all $s$. This uncertainty is related to  he quantum speed limit \cite{MT45,Bhattacharyya83}, for a time-independent Hamiltonian,
\begin{equation}
\begin{split}
    \abs{\braket{\psi(t)}{\psi(0)}} &\ge \cos^2 (\mathrm{std}_\psi H \cdot t) ~~~~~~0\le t \le \frac{\pi}{2 (\mathrm{std}_\psi H)}
\end{split}
\end{equation}

Hence, our Lemma \ref{lem:EscapeBoundsSuccess}  roughly states that if the angle between the initial ground state and the state of the system changes with slow $o(1/\poly(n))$ rate, the system  cannot reach the nearly orthogonal final ground state in polynomial time, thus failing the algorithm.

%
%
%
%

\textit{Theorem \ref{thm:main}b} uses Lemma \ref{lem:pert} to show that lowering $\theta$ from $n$ to $\theta_h$ has little effect on the system's path. Consider $H^{n}_s$ and $H^{\theta_h}_s$,  and their respective evolution operators $U^{n}(t)$ and $U^{\theta_h}(t)$. By the adiabatic theorem, for $\tau=n^4$, the path of the system propagating by $H^{n}_s$ is within distance $o(1)$ of the instantaneous ground state, but for simplicity we assume that they are equal, and leave the exact details to \footnote{See Supplementary Material at [URL will be inserted by publisher]}. The ground state ${\alpha^{n}_s}$ can be written as a tensor product:
\begin{equation}  \label{eq:ThetaNgs}
\begin{split}
\ket{\alpha^{n}_s}&=\prnt{\cos(\varphi_s)\ket{0}+\sin(\varphi_s)\ket{1}}^{\otimes n}\\&=\prnt{\cos(\xi_s)\ket{+}+\sin(\xi_s)\ket{-}}^{\otimes n}
\end{split}
\end{equation}
wherein, $\varphi_s\in[0,\pi/4]$ and $\xi_s=\pi/4-\varphi_s\in[0,\pi/4]$. Let $\delta H_s = H^{n}_s-H^{\theta_h}_s$. The integrand  in Eq. \ref{eq:pertH} is bounded by \begin{equation}
\begin{split}
\norm{{\delta H_s}\ket{\alpha^{n}_s}}
&\le h_{\max}\cdot  \max_s{\norm{\Pi_\mathcal{W} \ket{\alpha^{n}_s}}}
\\
h_{\max} &\triangleq \max_s(\norm{ \delta H_s}\le 2n
\end{split}
\end{equation}
The image of the perturbation $\delta H_s$ is the subspace $\mathcal W$, spanned by the union of states $\{\ket{j},\ket{j_+}\}$ with $h(j)>\theta_h$. 
 
The  projection of ${\alpha^{n}_s}$ only on the high-Hamming weight $\ket{j}$ states, squared,  is
\begin{equation} \label{eq:binomial}
\sum_{\substack{j \\h(j)\ge\theta_h}}\abs{\braket{j} {\alpha^{n}_s}}^2 = \sum_{k=\theta_h}^{n} {n \choose k} \cos^{2(n-k)}(\varphi_s)\sin^{2k}(\varphi_s)
\end{equation}
The overlap of ${\alpha^{n}_s}$ on the  high-Hamming weight $\ket{j_+}$ states is similar, with $\xi_s$ replacing $\varphi_s$. The RHS of Eq. \ref{eq:binomial} is the probability that the value of a binomial random variable  $X \sim {B}(n,p)$ is larger than $\theta_h$, with $p=\sin^2(\varphi_s)\le 1/2$. Since $\theta_h>n/2$ this probability is maximal when $p$ is maximal (${1}/{2}$). Using the Chernoff bound, 
\begin{equation} \label{eq:projectionWChernoff}
\begin{split}
&{\norm{\Pi_\mathcal{W} \ket{\alpha^{n}_s}}}^2 
\\
&\le
\sum_{\substack{j \\h(j)\ge\theta_h}}\abs{\braket{j} {\alpha^{n}_s}}^2+\abs{\braket{j_+} {\alpha^{n}_s}}^2
\\
&\le 2 \Pr (X\ge\theta_h) 
\le 2 e^{-\frac{(\theta_h-pn)^2}{3pn}}
\le 2 e^{-\frac{40\log n}{3}} = O(n^{-18}).
\end{split}
\end{equation}
Finally, by bounding the RHS of Eq. \ref{eq:pertH} we get
\begin{equation} \label{eq:boundingDisturbance}
\begin{split}
     \norm{(U^{n}(\tau)-U^{\theta_h}(\tau))\ket{\alpha^{n}_0}}  
     &\le
      \tau h_{\max} \max_s{\norm{\Pi_\mathcal{W} \ket{\alpha^{n}_s}}} 
      \\
      &\le n^4\cdot 2n\cdot O(n^{-9})
      =O(n^{-4}),
\end{split}
\end{equation}
which also bounds the influence of the perturbation $\delta H$ on the final state.


\subsection{Robustness}
The proof of Theorem \ref{thm:main}b reveals an interesting characteristic of  $H^{n}_s$. By  Eq. \ref{eq:projectionWChernoff},   the projection of $\alpha^{n}_s$ on $\mathcal W$ is $O(n^{-9})$, hence, by Eq. \ref{eq:boundingDisturbance}, \emph{any} perturbation $\delta H$ with $o(n^5)$ norm, whose image is $\mathcal W$, changes the  final state by $o(1)$ for $\tau=n^4$:
\begin{corollary} [\textbf{Spectrum perturbation robustness}] \label{cor:robustness}
Let $\widetilde H^{n}_0$ (respectively, $\widetilde H^{n}_1$) be equal to $H^{n}_0$ ($H^{n}_1$) except the energies of the states $\ket{j_+}$ ($\ket{j}$) with $h(j)>\theta_h$ are perturbed by $o(n^{5})$.  Evolving $\ket{\alpha_0^{n}}$  by $\widetilde H^{n}_s$ for $\tau=n^4$ yields a state within $o(1)$ distance of $\ket{\alpha_1^{n}}$. 
\end{corollary}

This corollary implies an unusual 
property of robustness in the adiabatic 
evolution with respect to $H^{n}_s$. 
Absurdly, one can choose a very large perturbation on the high Hamming-weight states, so that their energy is much {\it smaller} than that 
of the original ground state, and still this will change almost nothing in the evolution. 
In particular, one 
can assign negative energy to 
all states of at least $\theta_h$ minuses ($\ket{-}$) in $H_0^{n}$, 
and arrange it so that $\ket{+...+}$ becomes, e.g., the 10$^{\mathrm{th}}$ excited state of the initial Hamiltonian. Likewise, one can assign negative energy to states of 
at least $\theta_h$ 1's in $H_1^{n}$, such that $\ket{0...0}$ becomes, e.g.,  the 7$^{\mathrm{th}}$ excited state of the final Hamiltonian.  Yet, a system initialized to $\ket{+...+}$ would successfully reach $\ket{0...0}$.
Furthermore, by careful adjustment of these energies, the minimal gap may become arbitrarily small (e.g. zero at the end Hamiltonians), and the system would  still reach  $\ket{0...0}$\footnote{One can set the gap to be super-polynomially small for any $s$. Consider the Hamiltonian {$H^{n}_s + x\ketbra{1...1}{1...1}$}; when  {$x=O(n)$}, the perturbation {$x\ketbra{1...1}{1...1}$} have little influence on the evolution from {$\ket{+...+}$} by $H^{n}_s$. Fixing {$s$}, and considering an evolution  from {$H^{n}_s$} to {$\ketbra{1...1}{1...1}$} controlled by {$x$},  it is clear that the escape rate of  {$\mathcal V =\mathrm{span}\prnttt{\alpha^{n}_s}$} by the Hamiltonian {$\ketbra{1...1}{1...1}$} is {$o(1/\poly(n))$} small. Hence, we can find a {$x(s)$} for which the gap is super-polynomially small for all {$s$} (cf. \cite{FGGN08}).}.

The reason for this robustness is that the adiabatic path of the system is essentially confined to a subspace, 
orthogonal to that of the high-Hamming weight states (in both basis). 
The path is first limited to the symmetric subspace, whose dimension is $n+1$, because both the Hamiltonian and the initial ground state are symmetric to permuting the qubits. Inside the symmetric subspace the path of the system has very little overlap with high Hamming-weight states.


Such robustness phenomena might be useful in various contexts such as 
algorithm design and noise tolerance. 
We thus provide here a potentially 
interesting generalization: 
\begin{theorem} [\textbf{Robustness to bounded-rank perturbation}] \label{thm:robustnesss}
Let $\psi(t)$ be the wavefunction evolving by an $n$ qubit Hamiltonian $H(t)$, where $\norm{H(t)}\le h_{\max}$ for all $t\in[0,\tau]$. For any orthonormal basis of the Hilbert space, there exists a perturbation $\delta H(t)$ of rank $d\in \mathbbm N$, diagonalized in that basis, with bounded norm $g_{\max}$, s.t. evolving by $H(t)+\delta H(t)$ yields a state $\widetilde \psi(t)$, wherein
\begin{equation} \label{eq:theorem2}
\begin{split}
    \norm{\ket{\psi(\tau)} -\ket{\widetilde \psi(\tau)}} \le  \tau^2&{g_{\max}}^2 \prntt{(1-\eta) \sqrt{\frac{10    \tau h_{\max}  d}{2^n}}+\eta)}
    \\
    &\eta= \frac{1}{200}.
\end{split}
\end{equation}
\end{theorem}
Theorem \ref{thm:robustnesss} promises that for any quantum system evolving by a bounded-norm Hamiltonian (not necessarily adiabatic),  and for any basis to the Hilbert space, one can find some basis vectors, which have very little overlap with the path of the system. Therefore, adding a Hamiltonian perturbation acting on these vectors have very little influence on the dynamics of the system.

\begin{proof} (Of Theorem \ref{thm:robustnesss})
The proof relies on two parts: firstly, note that Corollary \ref{cor:robustness} can be generalized to any system path confined to a subspace of the Hilbert space (we denote this subspace by $\mathcal X$), as we show in the following Lemma \ref{lem:robustness}.
 Secondly, we show that a time $t$ adiabatic evolution of $n$ qubits, governed by $\poly(n)$-norm Hamiltonian, is confined to a polynomial (in $n$ and $t$) dimensional subspace $\mathcal X$; this is proven in Lemma \ref{lem:limitedsubspace}.

 \begin{lemma} [\textbf{Robustness by confinement}] \label{lem:robustness}
 Consider an evolution of a system by $H(t)$, where the path of the system is $\ket{\psi(t)}$.  Additionally, consider a subspace $\mathcal X$ confining the path:  $\norm{\Pi_{\mathcal{X}} \ket{\psi(t)}} \ge 1-\eta$,  for $t\in[0,\tau]$.
%
For any orthonormal basis of the Hilbert space, there exists a perturbation $\delta H(t)$ of rank $d\in \mathbbm N$, diagonalized in that basis, with bounded norm $g_{\max}$, s.t. evolving by $H(t)+\delta H(t)$  yields a state $\widetilde \psi(t)$, wherein
\begin{equation}
    \norm{\ket{\psi(\tau)} -\ket{\widetilde \psi(\tau)}} \le  \tau g_{\max} \prntt{(1-\eta) \sqrt{ 2^{-n} d \dim {\mathcal{X}}})+\eta}
\end{equation}
\end{lemma}
Lemma \ref{lem:robustness} can be interpreted as follows. Clearly, perturbations in $\mathcal X^\perp$ have no influence on the system's path and their rank may reach $\dim \mathcal H - \dim \mathcal X=2^n-\poly(n)$; however $\mathcal X$ is usually unknown. By Lemma \ref{thm:robustnesss}, given an \emph{arbitrary} basis, there is a perturbation of rank $\Omega(2^n/\poly(n))$, diagonalized in this basis, with negligible-influence on the system's path
\footnote{In retrospect, we could have proved Theorem \ref{thm:main}b using Theorem \ref{thm:robustnesss}. We have shown that  the higher the Hamming weight of the {$\ket{j},\ket{j_+}$} states, the least projection they have on {$\mathcal X$} (see Eq. \ref{eq:binomial}). Hence these states are the ideal candidates to perturb in Theorem \ref{thm:robustnesss}. The top {$2^n/\poly(n)$} highest Hamming weight states are the ones with Hamming weight  {$\ge n/2+O(\sqrt{n\log n})$}, which is in the order of {$\theta_h$}.
}. 

\begin{proof} (Of Lemma \ref{lem:robustness})
Consider a basis to the Hilbert space; on average, the projection of a random basis vector on $\mathcal X$ is neglectable if $\dim \mathcal X = O(\poly(n))$. By a probabilistic argument, for any basis there is at least a polynomial fraction of basis vectors, which have a very little overlap with the path of the system. Hamiltonian perturbations on the subspace $\mathcal W$ spanned by these vectors have a limited influence on  the system's path.  The above simple geometric statement can be 
stated 
formally as follows: \footnote{See Supplementary Material at [URL will be inserted by publisher]},
\begin{fact} \label{fact:LowInfluenceSubBasis}
Let $\prnttt{w_i}$ be a basis to a Hilbert space $\mathcal{H}$, and let ${\mathcal{X}}$ be a subspace in $\mathcal{H}$. For any $d\le\dim{\mathcal{H}}$ there exists a subspace ${\mathcal{W}}$ spanned by $d$ basis vectors of $\prnttt{w_i}$ s.t.  $\norm{\Pi_{\mathcal{W}} \Pi_{\mathcal{X}}}^2\le d \frac{\dim {\mathcal{X}} }{\dim {\mathcal H}}$.
\end{fact}
For any choice of basis (e.g., the basis diagonalizing  the initial/final Hamiltonians), there are $d\ll\frac{\dim \mathcal H}{\dim \mathcal X}$ basis vectors, spanning a subspace $\mathcal W$, which has a small projection on $\mathcal X$. By Lemma \ref{lem:pert}, Hamiltonian perturbations on $\mathcal W$ have a limited influence on the final state of the system.

By Fact \ref{fact:LowInfluenceSubBasis}, for any choice of basis, and for any $d$ there exists a subspace ${\mathcal{W}}$ of dimension $d$ s.t $\norm{\Pi_{\mathcal{W}}\Pi_{\mathcal{X}}}\le \sqrt{2^{-n}d \dim {\mathcal{X}}}$. By choosing $\delta H$ whose image is ${\mathcal{W}}$, the RHS of Eq. \ref{eq:pertH} can be bounded and the proof follows.
\end{proof}

The second part of the proof for Theorem \ref{thm:robustnesss} that a system evolving by $O(\poly(n))$ norm Hamiltonians for $O(\poly(n))$ time, covers approximately $O(\poly(n))$ dimensional subspace. 
\begin{lemma} [\textbf{Time induced confinement}] \label{lem:limitedsubspace}
Let $\psi(t)$ be the state of the system evolving by an $n$ qubit Hamiltonian $H(t)$, wherein $\norm{H(t)}\le h_{\max}$.  For any time $\tau>0$, there exists a subspace $\mathcal X$ of dimension $10\tau h_{\max}$, s.t.
\begin{equation}
    \norm{\Pi_{\mathcal X} \ket{\psi(t)} } \ge 1-\frac{1}{200}~~~~~ \forall t\in[0,\tau].
\end{equation}
\end{lemma}
\begin{proof}
Let $t_k = \frac{k}{10h_{\max}}$. The dimension of the subspace 
\begin{equation}
    \mathcal X = \mathrm{span}\prnt{\prnttt{\psi(t_k)|~ k=0,1,...10\tau h_{\max}-1}}
\end{equation} is at most $10\tau h_{\max}$, and by definition $\psi(t_k)\in \mathcal X$. It is left to lower-bound  the projection of $\psi(t)$ on $\mathcal X$ for  $t\neq t_k$. By Lemma \ref{lem:pert}, for the time interval $t \in [t_k,t_{k+1}]$, we get
\begin{equation}
\begin{split}
\norm{\ket{\psi(t)}-\ket{\psi(t_k)}} & \le \frac{1}{10},
\\
\abs{\braket{\psi(t)}{\psi(t_k)}}  & \ge \frac{1}{200}
\end{split}
\end{equation}
hence, for all $t\in [0,\tau]$,
\begin{equation}
\norm{\Pi_{\mathcal X} \ket{\psi(t)}} \ge   1- \frac{1}{200}.
\end{equation}
\end{proof}

\renewcommand{\qed}{\hfill$\blacksquare$}
The proof of Theorem \ref{thm:robustnesss} follows.


\end{proof}



\section{Discussion}


We have presented the $\theta$-Hamiltonian model and calculated the values of $\theta$ for which the final ground state is successfully found. Our analysis shows that a phase transition occurs at a surprisingly high point in the spectrum; moreover, energy states above the threshol have little or no effect on the evolution of the system to the extent that they can be considered to be in an orthogonal subspace (on the other hand, perturbing intermediate states with hamming weight between $\theta_\ell$ and $\theta_h$ can cause the gap to vanish). 

These results call for a more refined study of the dependence of the minimal gap on the spectra of both initial and final Hamiltonians; Such a study could be 
an important starting point towards 
the design of new adiabatic quantum algorithms, in less structured or symmetric settings. 

In order to prove our results, we have introduced the notion of escape rate and discussed its relation to the success of the evolution in finding the final ground state, and to the minimal gap. The adiabatic theorem provides one piece of the puzzle: it states that a large gap implies successful evolution. By Lemmas \ref{lem:EscapeBoundsSuccess}, \ref{lem:stuck_qw_aqc}, slow escape rates cause algorithms to fail, and, using the adiabatic theorem, infer super-polynomially small gap. On the other hand, Corollary \ref{cor:robustness} gives an example for a successful evolution even when the gap is super-polynomially small.
It seems that the notion of escape rate could thus be useful in situations where 
the minimal gap being large is too strict a condition or one which is too hard to prove. 

Our results reveal new interesting robustness properties in adiabatic 
evolutions. 
Theorem \ref{thm:robustnesss}  shows that the robustness to large perturbations in part of the spectrum of $H^{n}_0,H^{n}_1$ can appear in any bounded-norm time-dependent Hamiltonian regardless of symmetry. This may imply robustness of adiabatic evolutions to certain types of physical errors; also, this might allow some relaxation of the requirement on minimal gap when designing adiabatic quantum algorithms, because at least $1/\poly(n)$ fraction of the spectrum has little influence on the evolution of the system. Both directions remain to be explored.

Finally, and more technically, the problem of proving that for $\theta>\theta_h$ the gap is large (or maybe that it is not) remains open. This is interesting to clarify - and highlights yet again that it may be beneficial to study the success of adiabatic algorithms using a more refined tool than just the spectral gap.

\paragraph{Acknowledgments:}
\begin{acknowledgements}
The authors thank Zuzana Gavorova and Itay Hen
 for the helpful discussions. The work is supported by  ERC grant number 280157,  and Simons foundation grant number 385590. 
\end{acknowledgements}

\bibliography{bib}

\onecolumngrid
\newpage
\section{Supplemental Material}
 \subsection{Gap proof for $\theta=1,n$}
In this section we prove the minimal gap for $H^A_s$ and $H^B_s$ which corresponds to $\theta=1,n$ respectively:
\begin{claim}
The minimal gap of $H^A_s$ is $2^{-n/2}$
\end{claim}
\begin{proof}
Let $\ket{\alpha}=\ket{+\dots+}$, $\ket{\beta}=\ket{0\dots0}$. The Hamiltonian $H^A_s$ acts non-trivially on a two dimensional subspace spanned by $\alpha,\beta$ and be written as:
\begin{equation}
\begin{split}
&H^A_s=\mathbbm{1}_{2^n}-(1-s)\ketbra{\alpha}{\alpha}-s\ketbra{\beta}{\beta}=
\mathbbm{1}_{2^{n}-2}\oplus\prnt{
\begin{array}{cc}
s\prnt{1-\abs{\gamma}^2} & -s\gamma\sqrt{1-\abs{\gamma}^2} \\ 
-s\gamma^*\sqrt{1-\abs{\gamma}^2} & 1-s \prnt{1-\abs{\gamma}^2}
\end{array}}_{v_1,v_2},
\end{split}
\end{equation}
where the 2 dimensional matrix is written in the basis of  $\ket{v_1}=\ket{\alpha},\ket{v_2}=\frac{\ket{\beta}-\gamma\ket{\alpha}}{\sqrt{1-\abs{\gamma}^2}}$, and  $\gamma=\braket{\alpha}{\beta}$. In the subspace, the two eigenvalues are 
\begin{equation}
\lambda_{\pm}=\frac{1}{2}\pm\sqrt{\frac{1}{4}-(s-s^2)(1-\abs{\gamma}^2) }.
\end{equation}
The gap is minimum at $s=\frac{1}{2}$, where $\lambda_+-\lambda_-=\abs{\gamma}=\abs{\braket{\alpha}{\beta}}$. Since $\ket{\alpha}=\ket{+...+}$ and $ \ket{\beta}=\ket{0...0}$, the minimal gap equals $2^{-n/2}$. 

\end{proof}

\begin{claim}
The minimal gap of $H^B_s$ is $1/\sqrt{2}$.
\end{claim}
\begin{proof}
The tensor product form of $H^B$ simplifies the computation:
\begin{equation}
\begin{split}
H^B_s&=\sum_{k=1}^n \prntt{\prnt{1-s}\ketbra{-}{-}+s\ketbra{1}{1}}_k
=\sum_{k=1}^n\prnt{\begin{array}{cc}
(1-s)/2 & -(1-s)/2 \\ 
-(1-s)/2 & (1+s)/2
\end{array} }_k
\end{split}
\end{equation} 
where $k$ is the qubit's index. The eigenvalue of each two dimensional matrix are in the form  $\lambda_\pm=\frac{1 \pm \sqrt{1-2(s-s^2)}}{2}$, therefore the eigenvalues of $H^B$ are in the form $\lambda_\ell=(n-\ell)\lambda_++\ell\lambda_-$, where $\ell\in\{0,1,...n\}$. The minimal gap (denoted $\Delta$) between the ground state and the first excited state is at $s=1/2$:
\begin{equation}
\Delta (H^B_s) = \min_s (\lambda_1-\lambda_0) = \min_s ( \lambda_+-\lambda_-)=\frac{1}{\sqrt{2}}
\end{equation}
\end{proof}
\setcounter{lemma}{0}
\subsection{Proof of Lemma \ref{lem:EscapeBoundsSuccess}}

\begin{lemma} [\textbf{Escape rate bounds success}]
Let $H_s$ be an $n$ qubit Hamiltonian evolution by interpolation s.t. $\norm{H_s}=O(\poly(n))$.
 Additionally, let ${\mathcal{V}}$ be a subspace spanned by eigenstates  of $H_0$ whose projection $\Pi_{\mathcal{V}}$ satisfies $\norm{\Pi_{\mathcal{V}}\ket{\alpha_0}}=1$, and $\norm{\Pi_{\mathcal{V}}\ket{\alpha_1}}=o(1/\poly(n))$.  When adiabating $\alpha_0$ from $H_0$ to $H_1$ in time $\tau=O(\poly(n))$, while the escape rate of $\mathcal V$ by $H_1$ is $\beta<\pi/2\tau$, the projection of the final state on $\alpha_1$ is $\le  \sin(\beta\tau)+o(1/\poly(n))$.
\end{lemma}

\begin{proof}
The idea of the proof is to partition the propagator of the system to infinitesimal pieces, and to bound the change of the amplitude of the system on $\mathcal V$ by each piece. Let $U(\tau)$ be the unitary matrix applied to the system by the adiabatic evolution $H(t)=(1-\frac{t}{\tau})H_{0}+\frac{t}{\tau}H_{1}$ running from $t=0$ to $t=\tau$.  The propagator $U(\tau)$ can be written as product of $r$ unitary matrices, 
\begin{equation}
\begin{split}
U(\tau)&= U_r \cdots U_1
\\
U_j&=e^{-i\frac{\tau}{r}\prnt{1-\frac{j}{r}}H_{0}}\cdot e^{-i\frac{\tau}{r} \prnt{\frac{j}{r}}H_{1}}
\end{split}
\end{equation}
where $r\rightarrow \infty$. One can write the state of the system at time $t$ as follows  (up to a global phase):
\begin{equation}
    \ket{\psi(t)}=\cos(\varphi(t)) \ket{v(t)}+ \sin(\varphi(t)) \ket{\phi(t)},
\end{equation}
wherein $v(t)$ is a state in $\mathcal V$ and $\phi(t)$ is a state in $\mathcal V^\perp$.  The angle $\varphi(t)$ grows as the amplitude of the state of the system in ${\mathcal{V}}$ diminishes. We bound the angle $\varphi(\tau)$:
\begin{equation} \label{eq:cos_theta_j}
\begin{split}
\cos(\varphi(t+dt))&= \norm{\Pi_{\mathcal{V}}\ket{\psi(t+dt)}} 
=
\norm{\Pi_{\mathcal V}e^{-idt H(t)}\ket{\psi(t)}}
=
\norm{\Pi_{\mathcal{V}} e^{-idt \prnt{1-t/\tau}H_{0}}\cdot e^{-i dt\prnt{\frac{t}{\tau}}H_{1}}\ket{\psi(t)}}
\\
&=
\norm{\Pi_{\mathcal V}e^{-id t\prnt{\frac{t}{\tau}}H_{1}}\ket{\psi(t)}}
=
 \norm{\Pi_{\mathcal{V}} e^{-id t \prnt{\frac{t}{\tau}}H_{1}} \prntt{\cos(\varphi(t)) \ket{v(t)}+ \sin(\varphi(t)) \ket{\phi(t)}}}
 \\
&\ge 
\cos(\varphi(t))\norm{\Pi_{\mathcal{V}} e^{-i dt\prnt{\frac{t}{\tau}}H_{1}}  \ket{v(t)}} -  \sin(\varphi(t)) \norm{\Pi_{\mathcal{V}} e^{-i dt\prnt{\frac{t}{\tau}}H_{1}} \ket{\phi(t)}}
\end{split}
\end{equation}
By the definition of escape rate, we get the following two inequalities:
\begin{equation}
    \norm{\Pi_{\mathcal{V}} e^{-i dt\prnt{\frac{t}{\tau}}H_{1}}\ket{v(t)}}=\sqrt{1-\norm{\Pi_{\mathcal{V^\perp}} e^{-i dt\prnt{\frac{t}{\tau}}H_{1}}\ket{v(t)}}^2} 
    \ge \sqrt{1-\prntt{\beta dt+O(dt^2)}^2} = \cos(\beta dt) + O(dt^3)
\end{equation}
\begin{equation}
    \norm{\Pi_{\mathcal{V}} e^{-i dt\prnt{\frac{t}{\tau}}H_{1}} \ket{\phi(t)}}^2
    = \sum_{v_i\in\mathcal V} \abs{\bra{v_i}e^{-i dt\prnt{\frac{t}{\tau}}H_{1}}  \ket{\phi(t)} }^2 
    \le \beta dt + O(dt^2)=\sin (\beta dt)+ O(dt^2)
\end{equation}
Substituting back to Eq. \ref{eq:cos_theta_j}, we get
\begin{equation}
    \cos(\varphi(t+dt)) \ge \cos(\varphi(t)) \cos(\beta dt) -\sin(\varphi(t)) \sin(\beta dt) = \cos(\varphi(t)+\beta t)
\end{equation}
Hence, by recursion, we get $\cos(\varphi(\tau))\ge \cos (\tau \beta)$
\begin{equation}
\begin{split}
\norm{\Pi_{\mathcal{V}} \ket{\psi(\tau)}}&=\cos(\varphi(\tau)) \ge \cos(\beta\tau)
\\
\norm{\Pi_{\mathcal{V^\perp}} \ket{\psi(\tau)}}&= \sin(\varphi(\tau))\le \sin(\beta\tau).
\end{split} 
\end{equation}
In conclusion:
\begin{equation}
\begin{split}
    \abs{\bra{\alpha_1}U(\tau)\ket{\alpha_0}} = \abs{\braket{\alpha_1}{\psi(\tau)}} 
    &=
    {\norm{\Pi_{\mathcal{V}} \ket{\psi(\tau)}}\norm{\Pi_{\mathcal{V}} \ket{\alpha_1}}+\norm{\Pi_{\mathcal{V^\perp}} \ket{\psi(\tau)}}\norm{\Pi_{\mathcal{V^\perp}} \ket{\alpha_1}}} 
    \\
    &\le \sin(\beta\tau) + o(1/\poly(n)).
\end{split}
\end{equation}

\end{proof}

\subsection{Proof of Lemma \ref{lem:stuck_qw_aqc}}



\begin{lemma} [\textbf{Gap bound by escape rate}] 
Let $H_s$ be an $n$ qubit Hamiltonian evolution s.t. $\norm{H_s}=\poly(n)$. Additionally, let ${\mathcal{V}}$ be a subspace spanned by eigenstates of $H_0$, whose projection $\Pi_{\mathcal{V}}$ satisfies $\norm{\Pi_{\mathcal{V}} \ket{\alpha_0}}=1$  and $\norm{\Pi_{\mathcal{V}}\ket{\alpha_1}}\le 1/10$.
If the escape rate of the subspace $\mathcal V$ by $H_1$ is $\beta$, then 
 $\Delta \le \sqrt[4]{{100\beta}\norm{H_0-H_1}^3}$.
\end{lemma}


\begin{proof}
The idea of the proof is to assume by contradiction that for a specified $\tau$, the gap is large enough so that $\eta=1/10$ distance between $\alpha_1$ and the final state is reach. On the other hand, the escape rate bounds the rate by which the system leaves ${\mathcal{V}}$ (spanned by $\alpha_0$)  or, conversely, the rate by which the system reaches $\alpha_1$. Hence, if the escape rate is small enough so that the final state is not reached fast enough, the gap must have been small. We use the following version of the adiabatic theorem:
\begin{theorem} [Adiabatic theorem, adapted from \cite{AR04}] \label{thm:AdiabaticTheorem} 
Let $H_s=H_0 (1-s)+H_1s$, be a time dependent Hamiltonian $(0\le s\le 1)$, and let $\Delta$ be the minimal spectral gap between the  ground state and the excited states of $H_s$ for every $s$. Starting with the ground state of $H_0$, consider the adiabatic evolution given by $H_s$ applied for time $\tau$. Then, the following condition is enough to guarantee that the final state is at $\ell_2$ distance at most $\eta$ from the ground state of $H_1$:
\begin{equation}
\tau\ge \frac{10^4}{\eta^2} \cdot \max \prnttt{\frac{\norm{H'}^3}{\Delta^4},\frac{\norm{H'}\cdot\norm{H''}}{\Delta^3}},
\end{equation}
where prime denotes a derivative by $s$. 
\end{theorem}
By the adiabatic theorem, for a given $\tau$, if the gap is lower bounded as follows, then $\eta>1/10$:
\begin{equation} \label{eq:gapBound}
\Delta \ge \sqrt[4]{\frac{100\norm{H'}^3}{\tau}} =  \sqrt[4]{\frac{100\norm{H_0-H_1}^3}{\tau}}
\end{equation}

By Lemma \ref{lem:EscapeBoundsSuccess}, for $\beta<\pi/2\tau$, the projection of the final state on $\alpha_1$ is bounded by $\sin (\beta \tau)+o(1/\poly(n))$. Hence, using the same notations from the proof of Lemma \ref{lem:EscapeBoundsSuccess}, we get:
\begin{equation}
\begin{split}
1&\le \abs{\bra{\alpha_1}U(\tau)\ket{\alpha_0}} + \eta=\abs{\braket{\alpha_1}{\psi(\tau)}}+\eta=\abs{\cos(\varphi(\tau))\braket{\alpha_1}{v(\tau)}+\sin(\varphi(\tau))\braket{\alpha_1}{\phi(\tau)}}+\eta 
\\
&\le 0.1 + \sin(\beta\tau) + 0.1 + o(1/\poly(n)).
\end{split}
\end{equation}
By choosing $\tau=1/\beta$ we  get a contradiction - the system did not reach distance $\eta=1/10$ from $\alpha_1$. Hence, the gap lower bound in Eq. \ref{eq:gapBound} is incorrect and the proof follows.
\end{proof}

\subsection{Proof of Lemma \ref{lem:pert}}

\begin{lemma} [\textbf{Path shift by Hamiltonian perturbation}] 
Consider two time-dependent Hamiltonians $H(t), \widetilde H(t)$. Let $U(t)$ and $\widetilde U(t)$ be the respective evolution operators by these Hamiltonians  (e.g., when evolving by $H$, $\ket{\psi(t)}=U(t)\ket{\psi(0)}$). Then,
\begin{equation} \label{eq:pertHSM}
\norm{(U(\tau)-\widetilde U(\tau))\ket{\psi(0)}} \le \intop_0^\tau dt \norm{\prnt{H(t)-\widetilde H(t)}\ket{\psi(t)}}.
\end{equation}
\end{lemma}

\begin{proof}
Following Ambainis and Regev~\cite{AR04}, let $s$ be discretized to $r$ intervals  (later we'll take $r$ to $\infty$).  The unitary operation of the Hamiltonian evolution is discretized too: $U_j=\exp{\prnt{-i\frac{\tau}{r}H_{j/r}}}$ and $\widetilde{U}_j=\exp{\prnt{-i\frac{\tau}{r}\widetilde{H}_{j/r}}}$.

Let $\ket{\psi_j}=U_{j-1}\cdots U_0\ket{\psi(0)}$, and let $E_j=\widetilde{U}_j-U_j$. The distance between the two evolutions is the following (up to an $O(1/r)$ error):
\begin{equation} \label{eq:HtildeFinalState2} 
\begin{split}
&\norm{\ket{\psi_r} - \widetilde{U}_{r-1}\cdot \widetilde{U}_{r-2} \cdots \widetilde{U}_0 \ket{\psi_0}}
=
\norm{\ket{\psi_r} - \widetilde{U}_{r-1}\cdot \widetilde{U}_{r-2} \cdots \widetilde{U}_1 (U_0+E_0) \ket{\psi_0}}
\\
&=
\norm{\ket{\psi_r}- \widetilde{U}_{r-1}\cdot \widetilde{U}_{r-2} \cdots \widetilde{U}_1 (\ket{\psi_1}+E_0 \ket{\psi_0})}
=
\norm{\ket{\psi_r} - U_{r-1}\cdot U_{r-2} \cdots U_0 \ket{\psi_0}+ \sum_{j=0}^{r-1} \prnt{\widetilde{U}_{r-1}\cdots \widetilde{U}_{j+1}}E_j \ket{\psi_j}}
\\
&=
\norm{\sum_{j=0}^{r-1} \prnt{\widetilde{U}_{r-1}\cdots \widetilde{U}_{j+1}}E_j \ket{\psi_j}}
\le {\sum_{j=0}^{r-1} \norm{E_j\ket{\psi_j}}}  
\end{split}
\end{equation}
Approximating for $r\rightarrow\infty$:
\begin{equation} \label{eq:EjSM}
\begin{split}
{\sum_{j=0}^{r-1} \norm{E_j\ket{\psi_j}}}=& \sum_{j=0}^{r-1} \norm{{\prntt{{e^{-i\frac{\tau}{r}\widetilde{H}_{j/r}}}-e^{-i\frac{\tau}{r}H_{j/r}}}\ket{\psi_j}}}= \frac{\tau}{r} \sum_{j=0}^{r-1} \norm{{\prnt{H_{j/r}-\widetilde{H}_{j/r}}\ket{\psi_j}}}
\\
&=\intop_0^\tau dt \norm{\prnt{H(t)-\widetilde H(t)}\ket{\psi(t)}} 
\end{split}
\end{equation}
\end{proof}
\subsection{Proof of Theorem \ref{thm:main}}

\setcounter{theorem}{0}
\begin{theorem} [\textbf{Main: Phase transition}] ~ \\
\textbf{a.} Let $\theta_\ell\triangleq \frac{n}{2}-\sqrt{n\log^c n }$ for some constant $c>1$. If $\theta\le \theta_\ell$, then polynomial time adiabatic evolution by $H^{\theta}$ fails.  Furthermore, the minimal gap in this case is $o(1/\poly(n))$ small.\\
\textbf{b.} Let $\theta_h \triangleq n/2+\sqrt{40\log n}$. Evolving by  $H^{\theta}$, wherein $\theta \ge \theta_h$  succeeds  for  $\tau=n^4$.
\end{theorem}

The Chernoff bound is used in the both proofs:
let  $X_i$ be independent Bernoulli variables s.t. $\sum_i \mathbb{E} X_i= \mu$. Then,
\begin{gather}
\Pr\prnt{\sum_{i=1}^n X_i \ge (1+\delta)\mu}\le e^{-\frac{\delta^2\mu}{3}} ~~~~~~~~~ \delta \ge 0 \label{eq:Chernoff1}
\\
\Pr\prnt{\sum_{i=1}^n X_i \le (1-\delta)\mu}\le e^{-\frac{\delta^2\mu}{2}} ~~~~ 0\le \delta  \le 1 \label{eq:Chernoff2}
\end{gather}

 The following inequality is a derived from Eq. \ref{eq:Chernoff2}, with $\mu=n/2, \delta = 1-\theta/\mu$: 
\begin{equation}
2^{-n}\sum_{k=0}^{\theta-1} {n \choose k} \le e^{-\frac{(n/2-\theta)^2}{n}} ~~~~\theta\le n/2
\end{equation}

\begin{proof} (Theorem \ref{thm:main}a) ~\\
The proof relies on Lemma  \ref{lem:EscapeBoundsSuccess} and Lemma \ref{lem:stuck_qw_aqc}. We choose ${\mathcal{V}}$ to be the subspace spanned by the initial ground state $\alpha^{\theta}_0$ alone. The projection of $\alpha^{\theta}_1$ on $\mathcal V$ is $2^{-n/2}$, as required by both Lemmas.
We are left with the task of proving that the escape rate of $\mathcal V$ by  $H^{\theta}_1$ is $o(1/\poly(n))$ small.

Calculating the escape rate:
\begin{equation}
\begin{split}
    \norm{\Pi_{\mathcal{V^\perp}} H^{\theta}_1 \ket{\alpha^{\theta}_0}}^2 
    =
    \norm{H^{\theta}_1 \ket{\alpha^{\theta}_0}}^2-\norm{\Pi_{\mathcal{V}} H^{\theta}_1 \ket{\alpha^{\theta}_0}}^2
\end{split}
\end{equation}
The first addend of the RHS is:
\begin{equation}
    \norm{H^{\theta}_1 \ket{\alpha^{\theta}_0}}^2= \abs{\bra{0\dots 0}(H^{\theta}_1)^2 \ket{0\dots 0}} = \frac{1}{2^n}{\sum_{j=0}^{2^n-1} h^2_\theta(j)} \le \theta^2.
\end{equation}
The second addend is:
\begin{equation} 
\begin{split}
\norm{\Pi_{\mathcal{V}} H^{\theta}_1 \ket{\alpha^{\theta}_0}}
&=
\abs{\bra{0\dots 0}H^{\theta}_1 \ket{0\dots 0}}
=
\abs{\sum_{j=0}^{2^n-1} h_\theta(j) \abs{\braket{0\dots 0}{j_+}}^2} = \frac{1}{2^n}{\sum_{j=0}^{2^n-1} h_\theta(j)} = \frac{1}{2^n}{\sum_{k=0}^{n} {n \choose k} 
\min(k,\theta)} 
\\
&\ge
\theta - \frac{1}{2^n}\sum_{k<\theta} {n\choose k} (\theta-k) \ge \theta - \theta e^{-\frac{(n/2-\theta)^2}{n}}.
\end{split}
\end{equation}
Hence, 
\begin{equation}
\norm{\Pi_{\mathcal{V^\perp}} H^{\theta}_1 \ket{\alpha^{\theta}_0}}^2 \le \theta^2 - \prnt{\theta - \theta e^{-\frac{(n/2-\theta)^2}{n}}}^2 \le 2\theta^2  e^{-\frac{(n/2-\theta)^2}{n}}
\end{equation}
By choosing  $\theta \le n/2-\sqrt{n \log^c n}$ with $c>1$ we get that the escape rate is super-polynomially small and by Lemmas \ref{lem:EscapeBoundsSuccess},\ref{lem:stuck_qw_aqc}, polynomial time evolution fails, and the gap is o(1/\poly(n)).

\end{proof}
%

\begin{proof} (Theorem \ref{thm:main}b) \\
\noindent The idea of the proof is that evolving for polynomial time by  $H^{n}$ and by $H^{\theta}$, wherein $\theta\le  n/2 + \sqrt{40n\log n}$, yield two states with inverse polynomial  distance from each other.  Let  $U^{n}(\tau), U^{\theta}(\tau)$ be  the  evolution operators by $H^{n},H^{\theta}$ respectively.  The initial common ground state is $\alpha^{n}_0$. By the adiabatic theorem (Theorem \ref{thm:AdiabaticTheorem}) if $\tau=n^4$ the final state of the evolution by $H_s$ is $\alpha^{n}_1$ with additive error $o(1)$.

%
%
%

By Lemma \ref{lem:pert},
\begin{equation} 
\norm{(U^{n}(\tau)- U^{\theta}(\tau))\ket{\psi(0)}} \le \intop_0^\tau dt \norm{\prnt{H^{n}(t)- H^{\theta}(t)}\ket{\psi(t)}} = \tau \intop_0^1 ds \norm{\prnt{H^{n}_s-H^{\theta}_s}\ket{\psi_s}}
\end{equation}
where $\ket{\psi_s}=U^{n}(s\tau)\ket{\psi(0)}$. In  the rest of the proof we bound the integrand.
\begin{equation}
\label{eq:AQCerr}
\begin{split}
&\norm{\prnt{H^{n}_s-{H}^{\theta}_s}\ket{\psi_s}}=\\
&\norm{\prnt{1-s}\prnt{\sum_{k=0}^{2^n-1}\prnt{h(k)-{h}_\theta(k)}\ket{k}\bra{k} }\ket{\psi_s}+s\prnt{\sum_{k=0}^{2^n-1}\prnt{h(k)-{h}_\theta(k)}\ket{k_+}\bra{k_+} }\ket{\psi_s}}\le \\
&\norm{\sum_{k=0}^{2^n-1}\prnt{h(k)-{h}_\theta(k)}\ket{k}\braket{k} {\psi_s}}+\norm{\sum_{k=0}^{2^n-1}\prnt{h(k)-{h}_\theta(k)}\ket{k_+}\braket{k_+}{\psi_s}}.
\end{split}
\end{equation}

If the total evolution time $\tau$ is long enough, $\psi_s$ is not so far from $\alpha^{n}_s$, which is a tensor product and can be expressed in the following two forms:
\begin{equation}
\ket{\alpha^{n}_s}=(\cos(\varphi_s)\ket{0}+\sin(\varphi_s)\ket{1})^{\otimes n} =\prnt{\cos(\xi_s)\ket{+}+\sin(\xi_s)\ket{-}}^{\otimes n} 
\end{equation}
where $\varphi_s\in[0,\pi/4]$ and $\xi_s=\pi/4-\varphi_s\in[0,\pi/4]$. 

We write  $\psi_s$ as a similar tensor product by applying the adiabatic theorem (Theorem \ref{thm:AdiabaticTheorem}) on each qubit separately. The qubits evolves separately by a 1-qubit Hamiltonian with a constant minimal gap. Additionally, $\tau=n^4$, and $\norm{dH^{n}/ds}\le 2n$,  hence the distance for each 1-qubit final state from the ground state is $\eta =O(\tau^{-1/2})$. We get
\begin{equation}
\begin{split}
\ket{\psi_s}&=\prntt{\cos(\varphi_s)\ket{0}+\sin(\varphi_s)\ket{1}+O(1/\sqrt{\tau})}^{\otimes n}
\\
&= (\cos(\varphi_s')\ket{0}+e^{i\chi_s}\sin(\varphi_s')\ket{1})^{\otimes n}
=(\cos(\xi_s')\ket{+}+e^{i\kappa_s}\sin(\xi_s')\ket{-})^{\otimes n}
\end{split}
\end{equation}
where $\abs{\varphi_s-\varphi_s'}=O(1/\sqrt{\tau})$ and $\abs{\xi_s-\xi_s'}=O(1/\sqrt{\tau})$.
The square of first addend in the RHS of Eq. \ref{eq:AQCerr} takes the form: 

\begin{equation}
\begin{split}
\norm{\sum_{k=0}^{2^n-1}\prnt{h(k)-{h}_\theta(k)}\ket{k}\braket{k}{\psi_s}}^2
&\le 
\sum_{k=0}^{2^n-1}\prnt{h(k)-{h}_\theta(k)}^2\cos^{2n-2h(k)}(\varphi_s')\sin^{2h(k)}(\varphi_s')
\\
&\le n^2 \cdot 
 \sum_{k=\theta}^n  \binom{n}{k}\cos^{2n-2k}(\varphi_s')\sin^{2k}(\varphi_s'). 
\end{split}
\end{equation}
Note that in the last inequality $k$ is the Hamming weight. The second addend in  Eq.\ref{eq:AQCerr} is bounded by a similar expression by using $\xi_s$ instead of $\varphi_s$. Let  $p_s\triangleq \sin^2(\varphi_s')<1/\sqrt{2}+O(1/\sqrt{\tau})$. The sum is identical to the probability that $\theta$ or more out of $n$ unbiased coins will be head, when the probability for head is $p_s$. We will the Chernoff bound (Eq. \ref{eq:Chernoff1}) with $\delta=\frac{\theta}{\mu}-1$, $\mu=p_s n$ and $\theta=n/2 + \varepsilon$:

\begin{equation}
\begin{split}
\sum_{k=\theta}^n  \binom{n}{k}(1-p_s)^{n-k}p_s^{k} 
&\le
\exp\prnttt{-\prnt{\frac{\theta}{p_s n}-1}^2 \cdot \frac{p_s n}{3}} 
= \exp\prnttt{-\prnt{\frac{1}{2p_s}+\frac{\varepsilon}{np_s}-1}^2 \cdot \frac{p_sn}{3}} 
\\
&\le  \exp\prnttt{-\prnt{\frac{1}{1+O(1/\tau)}+\frac{2\varepsilon}{n}-1}^2 \cdot \frac{(\frac{1}{2}+O(1/\sqrt \tau)n)}{3}} 
= e^{-\frac{2\varepsilon^2}{3n} - O(\varepsilon/\sqrt{\tau}+n/\tau)}
\end{split}
\end{equation}
The sum is maximized for $p_s=1/2+O(1/\sqrt{\tau})$, hence the last inequality.  Substituting everything back to Lemma \ref{lem:pert}, we get: 
\begin{equation}
\norm{(U^{n}(\tau)- U^{\theta}(\tau))\ket{0...0}}  \le \tau \intop_0^1 ds \norm{{\prnt{H^{n}_s-{H}^{\theta}_s}\ket{\psi_s}}} \le 2 n \tau  e^{-\frac{\varepsilon^2}{3n} - O(\varepsilon/\sqrt{\tau}+n/\tau) } = 2n\tau O(n^{-9})
\end{equation}
For $\theta\le n/2+\varepsilon= n/2 + \sqrt{40n\log n}$, and $\tau=n^4$, the difference between the final state when evolving by $H^{\theta}_s$ and $\psi_1$ is $O(n^{-4})$ small. By the adiabatic theorem, the distance between $\psi_s$ and $\ket{0...0}$ is $o(1)$, and the proof follows.


%
%
%
%
%
%
\end{proof}

\subsection{Proof of Fact \ref{fact:LowInfluenceSubBasis}}
\setcounter{fact}{0}

\begin{fact} 
Let $\prnttt{w_i}$ be a basis to a Hilbert space $\mathcal{H}$, and let ${\mathcal{X}}$ be a subspace in $\mathcal{H}$. For any $d\le\dim{\mathcal{H}}$ there exists a subspace ${\mathcal{W}}$ spanned by $d$ basis vectors of $\prnttt{w_i}$ s.t.  $\norm{\Pi_{\mathcal{W}} \Pi_{\mathcal{X}}}^2\le d \frac{\dim {\mathcal{X}} }{\dim {\mathcal H}}$.
\end{fact}

\begin{proof}
The average projection of a basis element on ${\mathcal{X}}$ is: 
\begin{equation*}
\overline{ \norm{\Pi_{\mathcal{X}} \ket{w_i}}^2}=\frac{\sum_{i=1}^{\dim \mathcal H} \bra {w_i} \Pi_{\mathcal{X}} \ket{w_i}}{\dim \mathcal H} = \frac{\tr(\Pi_{\mathcal{X}})}{\dim \mathcal H}= \frac{\dim {\mathcal{X}}}{\dim \mathcal{H}} 
\end{equation*}
Let $D\subseteq \{1,2,...,\dim \mathcal H\}$, $\abs{D}=d$ and $\mathcal{W}=\mathrm{span} \{w_j\}_{j\in D}$.
\begin{equation}
\begin{split}
\norm{\Pi_{\mathcal{W}} \Pi_{\mathcal{X}}}^2
&
= \norm{\sum_{j\in D} \Pi_{\mathcal{X}} \ketbra{w_i}{w_i}}^2 
\le
\sum_{j\in D} \norm{\Pi_{\mathcal{X}} \ket{w_i}}^2
\end{split}
\end{equation}
One can always find $d$ vectors whose sum of squared projection (RHS) is at most $d$ times the average,  and the proof follows.
\end{proof}


\end{document}